\documentclass[12pt]{article}
\usepackage[utf8]{inputenc}

\usepackage{amsmath,amsfonts,amssymb,latexsym,amsthm,verbatim,a4wide,url,appendix}

 \theoremstyle{plain}

\newtheorem{theorem}{Theorem}[section]
\newtheorem{lemma}[theorem]{Lemma}
\newtheorem{proposition}[theorem]{Proposition}
\newtheorem{definition}[theorem]{Definition}
\newtheorem{remark}[theorem]{Remark}
\newtheorem{example}[theorem]{Example}
\newtheorem{corollary}[theorem]{Corollary}
\newtheorem{assumption}{Assumption}

\newcommand{\vd}[1]{\boldsymbol{#1}}
\newcommand{\vc}[1]{{\mathbf{#1}}}
\newcommand{\ep}{\mathbb{E}}
\newcommand{\pr}{\mathbb{P}}
\newcommand{\re}{\mathbb{R}}
\newcommand{\wt}[1]{\widetilde{#1}}
\newcommand{\var}{{\rm Var\;}}

\newcommand{\rhocorr}{\rho_{PCC}}
\newcommand{\LL}{L}
\newcommand{\Cn}{C_{(n)}}
\newcommand{\CnG}[2]{C_{(#1,#2)}}
\newcommand{\fn}[1]{#1^{(n)}}
\newcommand{\ft}[1]{#1^{(2)}}

\newcommand{\ff}[1]{#1^{(1)}}
\newcommand{\fm}[1]{#1^{(m)}}
\newcommand{\fnG}[3]{#1^{(#2,#3)}}
\newcommand{\Theb}[1]{\Theta^{(#1)}}
\newcommand{\TheG}[2]{\Theta^{(#1,#2)}}

\newcommand{\Jst}{J_{\rm st}}
\newcommand{\skewn}{\gamma_3}
\newcommand{\ol}[1]{\overline{#1}}

\title{Maximal correlation and the rate of Fisher information convergence in the Central Limit Theorem}
\author{Oliver Johnson}
\date{\today}

\begin{document}

\maketitle

\begin{abstract}
\noindent We consider the behaviour of the Fisher information of scaled sums of independent and identically distributed random variables in the Central Limit Theorem regime. We show how this behaviour can be related to the second-largest non-trivial eigenvalue of the operator associated with the Hirschfeld--Gebelein--R\'{e}nyi maximal correlation. We prove that assuming this eigenvalue satisfies a strict inequality, an $O(1/n)$ rate of convergence and a strengthened form of monotonicity hold.
\end{abstract}

\section{Introduction}

Consider independent and identically distributed (i.i.d.) random variables $Y_i \sim Y$ taking values in $\re$, with mean $0$ and variance $\sigma^2 < \infty$, and write $S_n := Y_1 + Y_2 + \ldots + Y_n$ for their sum. We assume that the $Y_i$ have smooth densities, and consider the behaviour of the Fisher information in the Central Limit Theorem regime.  

\begin{definition} \label{def:fisher}
For any random variable $X \in \re$ with absolutely continuous density $p_X$ we define the Fisher score function (with respect to location parameter) $\varrho_X(x) := p_X'(x)/p_X(x)$ and Fisher information $J(X) := \ep \varrho_X(X)^2 = \int p_X(x) \varrho_X(x)^2 dx$. Further, as in \cite{johnson5}, we write the standardized Fisher information (standardized Fisher divergence)
\begin{equation} \label{def:Jst} \Jst(X) := (\var X) \ep \left( \varrho_X(X) + \frac{X}{\var X} \right)^2 = (\var X) J(X) - 1.\end{equation}
\end{definition}

The quantity $\Jst(X)$ is scale--invariant, and is $(\var X)$ times the quantity sometimes referred to as Fisher divergence or as Fisher information distance. 
The non-negativity of $\Jst$ is equivalent to the standard Cram\'{e}r-Rao lower bound (see for example \cite[Eq. (2.1)]{stam}), with equality holding if and only if $X$ is Gaussian. Hence, if $\Jst(X)$ is `small', then intuitively $X$ should be `close to Gaussian'. In fact, controlling the standardized Fisher information gives a strong sense of convergence to Gaussian, with control of $\Jst$ implying control of total variation distance, Hellinger distance and the supremum distance between densities (see \cite[Lemma 1.5]{johnson5}) and relative entropy (see \cite[p409]{johnson5}). The fact that absolute continuity is a sufficient condition for the existence of Fisher information is discussed for example in \cite[Section 4.4]{huberrobust}.

We follow Courtade \cite{courtade2016} in analysing Fisher information using quantities related to the (Hirschfeld--Gebelein--R\'{e}nyi) maximal correlation $\rho_{\max}$ \cite{gebelein1941, hirschfeld1935,renyi1959}. It is well-known that the standard (Pearson) correlation coefficient $\rhocorr(U,V)$ only captures linear relationships between random variables, and hence can be zero even when $U$ and $V$ are dependent. In contrast, the maximal correlation between random variables $U$, $V$ is the largest correlation between non-constant well-behaved functions of them 
\begin{equation}
\rho_{\max}(U,V) := \sup_{f,g} \rhocorr \left( f(U), g(V) \right). \end{equation} Like the mutual information, $\rho_{\max}(U,V)$ is zero if and only if $U$ and $V$ are independent, see \cite{gebelein1941, hirschfeld1935,renyi1959}.  The maximal correlation has found application in information theory partly because of its relation to hypercontractivity and the strong data processing constant \cite{anantharam3,kamath}.

Courtade \cite{courtade2016} gave a direct and simple proof of the monotonicity of Fisher information in the Central Limit Theorem regime, using the fact that for i.i.d. $Y_i$ the maximal correlation between sums of different sizes satisfies 
\begin{equation} \label{eq:DKS} 
\rho_{\max}(S_m, S_n) = \sqrt{\frac{m}{n}} \mbox{\;\;\; for all $m \leq n$}. \end{equation} 
This fact, which we call the Dembo--Kagan--Shepp (DKS) identity \cite{DKS2001}, can be understood through an equivalent formulation of $\rho_{\max}$ as the largest non-trivial singular value of conditional expectation operators (see Section \ref{sec:condexpprop}). This identity was originally proved in \cite{DKS2001} under the assumption that $Y_i$ have finite variance, a condition subsequently relaxed in \cite{bryc2}. We note that
Courtade's proof \cite{courtade2016} of monotonicity via the DKS identity only recovers the result along i.i.d. sequences, which is less general than the `leave-one-out' inequality proved by Artstein, Ball, Barthe and Naor  \cite{artstein}.  However, Courtade has subsequently shown that many monotonicity results, including the DKS identity and the general subset inequalities of Madiman and Barron \cite{madiman} can be seen as immediate consequences of Shearer's lemma \cite{courtade2}.

In this paper we work with a quantity $\Theb{n}$ defined in terms of the second-largest non-trivial singular value of the same conditional expectation operators, defined in Definition \ref{def:higher} below, and satisfying $\Theb{n} \geq 0$ by the Dembo--Kagan--Shepp identity \cite{DKS2001}. Under a technical diagonalizability condition (Assumption \ref{ass:diag} below, which is assumed to hold throughout) a more detailed analysis of $\Theb{n}$ using the Efron--Stein (ANOVA) decomposition \cite{efronstein} allows us to deduce the following result:

\begin{theorem} \label{thm:main}
Consider i.i.d. $Y_i \sim Y$ with mean $0$ and variance $\sigma^2 < \infty$ and smooth densities on $\re$. For any $n$, writing $\Theb{2}$ for the quantity from Definition \ref{def:higher} below, then
\begin{equation} \label{eq:jstmain}
 \Jst \left( \frac{Y_1 + \ldots + Y_n}{\sqrt{n}} \right)
\leq  \frac{\Jst(Y)}{1 + \Theb{2} (n-1)}.
\end{equation}
In other words, if $\Theb{2} > 0$ then we achieve a $O(1/n)$ convergence rate of standardized Fisher information.
\end{theorem}

Theorem \ref{thm:main} follows directly by combining Propositions \ref{prop:fishbd1} and \ref{prop:DKShigher} below. Note that Artstein, Ball, Barthe and Naor \cite{artstein2} and Johnson and Barron \cite{johnson5} both proved an $O(1/n)$ rate of convergence of standardized Fisher information (and hence of relative entropy) for one-dimensional random variables assuming finiteness of the Poincar\'{e} constant (this was extended to the $\re^d$ case by \cite{ball} under a stronger assumption of log-concavity). However, since (see Lemma \ref{lem:weaker} below), finiteness of the Poincar\'{e} constant implies $\Theb{2} > 0$, we can regard our condition as weaker. As with Poincar\'{e} constants, positivity condition $\Theb{2} > 0$ implies finiteness of moments of all orders (see Proposition \ref{prop:R2bound} below). However, unlike finiteness of Poincar\'{e} constants, positivity of $\Theb{2}$ does not directly require that the support of $Y$ is connected.

 To illustrate the relationship between moments and $\Jst$, we further prove  a lower bound on the Fisher information which tightens the lower bound of \cite[Lemma 1.4]{johnson5}, and which complements the upper bound in Theorem \ref{thm:main}:

\begin{lemma} \label{lem:fishbd3}
For i.i.d. $Y_1, \ldots, Y_n \sim Y$ the standardized Fisher information satisfies
\begin{equation} \Jst(U_n) \geq \frac{\skewn^2}{\Sigma + 2(n-1)}, \label{eq:fishLB} \end{equation}
where $\skewn = \ep Y^3/\sigma^3$ is the skewness of $Y$ and $\Sigma = \ep Y^4/\sigma^4 - (\ep Y^3/\sigma^3)^2 - 1 \geq 0$.
\end{lemma}

The upper and lower bounds on $\Jst(U_n)$ given by \eqref{eq:jstmain} and \eqref{eq:fishLB} are compatible in the sense that (since by \eqref{eq:R2LB} the $\Theb{2} \leq \frac{2}{\Sigma}$) we know
$$\frac{\Jst(Y)}{1 + \Theb{2} (n-1)} \geq \frac{\Jst(Y)}{1 + (2/\Sigma)(n-1)} 
= \frac{\Sigma \Jst(Y)}{\Sigma + 2 (n-1)} \geq  \frac{\skewn^2}{\Sigma + 2 (n-1)},
$$
where the final inequality simply follows from the case $n=1$ of \eqref{eq:fishLB}.

The need for finiteness of the Poincar\'{e} constant to ensure $O(1/n)$ convergence of Fisher information and of relative entropy was removed in subsequent work of Bobkov, Chistyakov and G\"{o}tze (see for example \cite{bobkov2014fisher} for Fisher information and \cite{bobkov6, bobkov8} for relative entropy). These papers proved this rate of convergence under the assumption of finite fourth moment, as well as a variety of related results under a moment-matching assumption. Note that (by Lemma \ref{lem:RnExample} below) if the fourth moment is infinite, our methods do not give $O(1/n)$ convergence, so our results should be regarded as weaker. However, papers \cite{bobkov6, bobkov8} used a detailed argument involving Edgeworth expansions, truncation of densities and analysis of the characteristic function to derive their results. We believe our results are obtained in a more straightforward way, and the connection to maximal correlation in this context may be of independent interest. Further, we prove a novel strengthened form of monotonicity, Theorem \ref{thm:main2}, which places monotonicity and convergence results in the same framework, whereas they have often historically been treated separately.

An alternative perspective was provided by Courtade, Fathi and Pananjady \cite{courtade3}, who weakened the Poincar\'{e} constant assumption to require only the existence of a Stein kernel $\tau$ (which holds for any centered random variable with connected support). Using this, they proved an $O(1/n)$ rate of convergence in Wasserstein distance and an $O(\log n/n)$ rate of convergence in relative entropy, with the speed of convergence being dictated by the Stein discrepancy (squared distance from the Stein kernel $\tau$ to the identity). This work has the considerable advantage of holding in more general spaces $\re^d$ for $d \geq 1$. It would be of interest to understand the relationship between our $\Theb{2} > 0$ condition and the Stein condition of \cite{courtade3}.

The problem of proving information--theoretic versions of the Central Limit Theorem is a long-standing one, the early history of which is reviewed in \cite{johnson14}. In particular, we mention work of Linnik \cite{linnik} and Shimizu \cite{shimizu}. However, our work follows the idea of studying projections of score functions, and follows a path first set out by Stam \cite{stam}, Brown \cite{brown}, Barron \cite{barron}, as well as exploiting subsequent developments. In particular, the analysis of \cite{johnson5} exploited the fact that in the limit the score function of the limit must simultaneously be both a ridge function (a function $f(x_1+ \ldots + x_n)$) and close to being the sum $f_1(x_1) + \ldots + f_n(x_n)$, and hence must be close to being linear.

This analysis generalized a key step  in the work of Brown (and later in Barron \cite{barron}), which was an inequality \cite[Lemma 3.1]{brown} concerning properties of Hermite polynomials, which are orthogonal in the Gaussian case. Our work can be seen as giving an alternative generalization of this, using an orthogonal function expansion based on the Singular Value Decomposition.  The link between these two ideas is the fact that the Hermite polynomials provide the Singular Value Decomposition of conditional expectations in the Gaussian case  (see \cite[Theorem 3]{makur} and Example \ref{ex:hermite}).

The structure of the remainder of the paper is as follows: in Section \ref{sec:condexpprop} we formally define the conditional expectation operators and the eigenvalue--related quantity $\Theb{n}$. In Section \ref{sec:condexpprop2} we give examples where we can calculate $\Theb{n}$ explicitly, discuss properties of $\Theb{n}$ and show how it relates to other quantities. In Section \ref{sec:scoreconv} we discuss how standard results allow us to control the value of the standardized Fisher information $\Jst$ on convolution, in terms of $\Theb{n}$. In Section \ref{sec:DKShigher} we discuss how to control higher order terms in the Dembo--Kagan--Shepp argument, and hence bound $\Theb{n}$ in terms of $\Theb{2}$. In Section \ref{sec:strengthen} we show how these arguments imply a stronger form of monotonicity of Fisher information. We conclude with some suggestions for future work in Section \ref{sec:futurework}. 

\section{Conditional expectation operator definitions} \label{sec:condexpprop}

We introduce notation based on \cite{makur}. For any probability measure $\pr$ we write $L^2(\pr)$ for the Hilbert space endowed with inner product
$\langle f, g \rangle_{\pr} = \int f(x) g(x) d\pr(x)$.
 Write $\pr_{Y}$ and $\pr_{S_n}$ for the law of the relevant random variables, where as before $S_n = Y_1 + \ldots + Y_n$.
\begin{definition}
Define conditional expectation operator $\Cn: \LL^2(\pr_{Y}) \mapsto \LL^2(\pr_{S_n})$ and its adjoint $\Cn^*: \LL^2(\pr_{S_n}) \mapsto \LL^2(\pr_{Y})$ by:
\begin{eqnarray}
(\Cn f)(s) & = & \ep \left[ f(Y_1) | S_n = s \right], \label{eq:Cndef} \\
(\Cn^* g)(y) & = & \ep \left[ g(S_n) | Y_1 = y \right]. \label{eq:CnStdef}
\end{eqnarray}
\end{definition}

These maps are adjoint in the sense that (by direct calculation, or the tower law) for all $f$ and $g$: 
\begin{equation} \label{eq:adjoint}
\left\langle g, \Cn f \right\rangle_{\pr_{S_n}} = \left\langle \Cn^* g, f \right\rangle_{\pr_{Y}} = \ep \left[ f(Y_i) g(S_n) \right]. \end{equation}

\begin{assumption} \label{ass:diag}
We assume throughout this paper that the self-adjoint map $\Cn^* \Cn$ is diagonalizable. 
\end{assumption}

\begin{definition} \label{def:eigenfunctions}
Under Assumption \ref{ass:diag} write $(\ff{f}_k)_{k = \{0,1, \ldots \}} \in \LL^2(\pr_{Y})$ for the basis of orthonormal eigenfunctions of $\Cn^* \Cn$,
 with corresponding eigenvalues $\fn{\lambda}_k$ and singular values
$\fn{\mu}_k = \sqrt{ \fn{\lambda}_k}$. Here, without loss of generality, we assume that 
$$ 1 = \fn{\lambda}_0 \geq \fn{\lambda}_1 \geq \ldots \geq  0.$$
We write $\fn{g}_k = (\Cn \ff{f}_k)/\fn{\mu}_k$ for the scaled images of these eigenfunctions. \end{definition}

\begin{remark} \label{rem:eigen} \mbox{ }
\begin{enumerate}
    \item Note that by \eqref{eq:adjoint} the functions $\fn{g}_k$ are orthonormal in $\LL^2(\pr_{S_n})$. 
Further, note that 
\begin{equation} \Cn^* ( \fn{g}_k)  = \frac{1}{\fn{\mu}_k} (\Cn^*  \Cn \ff{f}_k)
= \fn{\mu}_k \ff{f}_k. \label{eq:gimage} \end{equation}
\item \label{it:eigen1}
Note that $\ff{f}_0 = \fn{g}_0 \equiv 1$ and the pair $(\ff{f}_1 , \fn{g}_1)$ achieves the maximum correlation since by \eqref{eq:adjoint} we know
$$ \ep \left( \ff{f}_1(Y_i) \fn{g}_1(S) \right) = \left\langle \fn{g}_1 , (\Cn \ff{f}_1) \right\rangle_{\pr_S} = 
\left\langle \fn{g}_1 , \fn{\mu}_1 \fn{g}_1 \right\rangle_{\pr_S} = \fn{\mu}_1.$$
\item \label{it:iid}
In this i.i.d. case, we can take $\ff{f}_1(y) = y/\sigma$, $\fn{g}_1(s) = s/(\sigma \sqrt{n})$ with $\fn{\mu}_1 = 1/\sqrt{n}$.
 This choice of functions has the relevant properties since by symmetry (or the fact that averages of i.i.d. random variables form a reverse martingale)
$$ (\Cn \ff{f}_1)(s) = \frac{1}{\sigma} \ep( Y_1 | S_n = s) = \frac{s}{\sigma n}
= \fn{\mu}_1 \fn{g}_1(s),$$
and
$$ (\Cn^* \fn{g}_1)(y) = \frac{1}{\sigma \sqrt{n}} \ep( y + Y_2 + \ldots + Y_n) = \frac{1}{\sigma \sqrt{n}} y = \fn{\mu}_1 \ff{f}_1(y).$$
The DKS identity \eqref{eq:DKS} tells us that no larger value of $\fn{\mu}_1$ is possible.
\end{enumerate}
\end{remark}
The focus of this paper will be the quantity $\Theb{n}$  defined  in terms of the 
second-highest non-trivial eigenvalue of the self-adjoint map $\Cn^* \Cn$ as:

\begin{definition} Using the notation above, write
\label{def:higher}
\begin{equation} \label{eq:higher}
 \Theb{n} :=   \frac{1}{n \fn{\lambda}_2} - 1 =
 \inf_{h: \ep h(S_n) = \ep S_n h(S_n) = 0} \frac{1}{n} \frac{\ep \left( h(S_n)^2 \right)}{ \ep 
\left( (\Cn^* h) (Y)^2\right)}  - 1.   \end{equation}
\end{definition}
The Dembo--Kagan--Shepp identity \cite{DKS2001} means that for $k \geq 2$, eigenvalues $\fn{\lambda}_k$ are $\leq 1/n$, which ensures that $\Theb{n} \geq 0$.
While we are not aware of existing results in the literature that bound $\fn{\lambda}_k$ for $k \geq 2$, we remark that the higher order eigenfunctions $\ff{f}_k$ and $\fn{g}_k$ (for $1 \leq k \leq K$, for some fixed $K$) have been used in a manner similar to Principal Components Analysis to capture significant high-order features of datasets \cite{makur2015efficient}.

One possible strategy to show that $\Theta^{(2)} > 0$ is to show that $\Cn$ and $\Cn^*$ are compact operators (recall from e.g. \cite[Section 3.1]{simon2015} that a compact linear operator is one for which the image of any bounded subset has compact closure). The Riesz--Schauder Theorem  \cite[Theorem 3.3.1]{simon2015} states that the only possible accumulation point of eigenvalues of a compact operator is at 0, so if the eigenspace corresponding to $\lambda = 1/2$ has dimension 1 then we can deduce $\ft{\lambda}_2 < 1/2$. We consider the second point in Remark \ref{rem:unique} below, and discuss the question of compactness now. 

This compactness is stated as \cite[Assumption 5.2]{breiman1985}, which states that it `is satisfied in most cases of interest' and in particular if a sufficient condition \cite[Eq. (5.4)]{breiman1985} holds -- we derive this condition here for completeness. As in \cite[Eq. (40]{makur} we can expand the Radon-Nikodym derivative between joint and marginal densities using the Singular Value Decomposition as:
\begin{equation} 
\tau_{n}(y,s) := \frac{p_{Y_1,S_n}(y,s)}{p_{Y_1}(y) p_{S_n}(s)} = \sum_{k=0}^\infty \fn{\mu}_k \ff{f}_k(y) \fn{g}_k(s). 
\end{equation}
Note that $(\Cn f)(s) = \int \tau_n(z,s) p_Y(z) f(z) dz$ and $(\Cn^* g)(y) = \int \tau_n(y,s) p_{S_n}(s) g(s) ds$, so
\begin{equation}
\left( \Cn^* \Cn f \right)(y) = \int f(z) p_Y(z) L_n(z,y) dz,
\end{equation}
where $L_n(z,y) := \int p_{S_n}(s) \tau_n(z,s) \tau_n(y,s)  ds$ is  symmetric, as expected. Then, $(\Cn^* \Cn)$ is compact if this is a trace-class operator (see \cite[Section 3.6]{simon2015}), or in other words that we can use Mercer's Theorem (\cite[Theorem 3.11.9]{simon2015}) to verify that
\begin{equation}
T_n(Y) := \int p_Y(y) L_n(y,y) dy  = \iint p_Y(y) p_{S_n}(s) \tau_n(y,s)^2  dy ds < \infty
\end{equation}
(this is \cite[Eq. (5.4)]{breiman1985}). 

Note that this quantity $T_n(Y)$ has the property that $T_n(Y) - 1 = D_{\chi^2}( p_{Y_1, S_n} \| p_{Y} \times p_{S_n})$,
where  $Y$ is an independent copy of $Y_1$ and
we write $D_{\chi^2}(f \| g) = \int (f(x)/g(x) -1)^2 g(x) dx= \int f(x)^2/g(x) dx - 1$ for the $\chi^2$-divergence.
Using this, an anonymous referee provided a prood of the following theorem, which shows that arbitrarily small Gaussian regularizations of sub-Gaussian random variables have the trace-class property:

\begin{theorem} \label{thm:subgauss}
For any $\delta > 0$, taking $Z$ Gaussian with mean $0$ and variance $\delta^2$ and $Y \sim X + Z$,
then writing $X'$ for an independent copy of $X$:
$$ D_{\chi^2}( p_{Y_1, S_n} \| p_{Y} \times p_{S_n}) \leq
 \frac{1}{1-1/n} \ep \left( \exp \left( \frac{ (X_1 - X_1')^2}{(n-1) \delta^2} \right) \right).$$
Hence if $X$ is sub-Gaussian then  for $n$ sufficiently large $ T_n(Y) < \infty$ and hence $(\Cn^* \Cn)$ is compact.
\end{theorem}
\begin{proof} See Appendix \ref{sec:subgausspf}. \end{proof}

We briefly mention that by linearizing the logarithm, we can bound the mutual information $I(Y_1; S_n) = D( P_{Y_1,S_n} \| P_{Y_1} P_{S_n}) = \int p_{Y_1}(y) p_{S_n}(s) \tau_n(y,s) \ln \left( \tau_n(y,s) \right) dy ds \leq T_n(Y)- 1$. Further, 
$I(Y_1; S_n) = H(Y_1) + H(S) - H(Y_1,S_n) = H(Y_1) + H(S_n) - H(Y_1,Y_2) =  H(S_n) - H(Y) \geq (\log 2)/2$,
where the lower bound on $H(S_n) - H(Y)$ follows by Shannon's Entropy Power Inequality. In other words, finiteness of $T_2 := T_2(Y)$ ensures a reverse Entropy Power Inequality of the form $\exp(2 h(S_n)) \leq C \exp(2 h(Y)) $, where $C = \exp( 2(T_2-1))$.

\begin{example}
In the case where $p \sim N(0,1)$, we can explicitly write down $\tau_{2}(y,s) = \sqrt{2}
\exp \left( -\frac{1}{4} ( s^2 - 4 s y + 2 y^2) \right)$, and direct calculation gives that
$$ T_2 =  \iint p_Y(y) p_{S_2}(s) \tau_{2}(y,s)^2  dy ds = 2.$$
This confirms the values in Example \ref{ex:hermite} below, which gives that the eigenvalues are $\ft{\lambda}_k = 2^{-k}$, and so $\sum_{k=0}^\infty \ft{\lambda}_k = 2$, confirming the value of the trace by Lidskii's Theorem \cite[Corollary 3.12.3]{simon2015}.
\end{example}

\begin{remark} This formulation gives an alternative proof of the Dembo--Kagan--Shepp identity for $n=2$, using the fact that $\tau_{2}(z,s) p_Y(z) = p_Y(z) p_Y(s-z) = 
\tau_{2}(s-z,s) p_Y(s-z)$. Fix $s$ and for function $f(z)$, write $\ol{f}(z) = f(s-z)$. Then
$$ \Cn \ol{f}(s) = \int \tau_{2}(z,s) p_Y(z) \ol{f}(z) dz
= \int \tau_{2}(s-z,s) p_Y(s-z) f(s-z) dz = \Cn f(s).$$
Hence, for any $f$ with $\int f(z) p_Y(z) dz = 0$, Cauchy-Schwarz gives
\begin{eqnarray}
 \left( 2 \Cn f(s) \right)^2 
& = &  \left( \int \tau_{2}(z,s) p_Y(z) \left( f(z) + \ol{f}(z) \right) dz \right)^2 \nonumber \\
& \leq & \left( \int \tau_{2}(z,s) p_Y(z) \left( f(z) + \ol{f}(z) \right)^2 dz \right) \left(  \int \tau_{2}(z,s) p_Y(z) dz \right) \label{eq:CSstep} \\
& = & \left( \int \tau_{2}(z,s) p_Y(z) \left( f(z) + \ol{f}(z) \right)^2 dz \right), \nonumber
\end{eqnarray}
since $\int p_Y(z) \tau_{2}(z,s) dz \equiv 1$. The result follows on multiplying by $p_S(s)$ and integrating, to deduce that
$4 \ep (\Cn f(S))^2 \leq 2 \ep f(Y)^2$, or $\lambda^{(2)}_2  \leq 1/2$.

Note that (see also Remark \ref{rem:unique} below)  equality holds in \eqref{eq:CSstep} if and only if $f(z) + \ol{f}(z)$ is constant in $z$. Taking a derivative with respect to $z$, we deduce that $f'$ must be constant, or that $f$ linear is the unique eigenfunction achieving $\lambda = 1/2$.
\end{remark}

\section{Conditional expectation operator properties} \label{sec:condexpprop2}

We now review two examples where we can explicitly calculate the eigenfunctions and eigenvalues of $\Cn^* \Cn$, using properties of orthogonal polynomials \cite{abramowitz}, and hence deduce the value of $\Theb{n}$. Instead of orthogonal polynomials, these calculations can alternatively be performed using properties of the associated semigroups (Ornstein--Uhlenbeck and Laguerre semigroups, respectively). First, the Gaussian case (see also \cite{makur}):

\begin{example} \label{ex:hermite}
If $Y_i$ are Gaussian with variance $\sigma^2$, then $\ff{f}$ and $\fn{g}$ are orthonormalized Hermite polynomials. For any $\alpha$ we define $H_n^{(\alpha)}(x) = H_n(x/\sqrt{\alpha})$ (where $H_n$ are the Hermite polynomials, which are orthogonal with respect to standard Gaussian weights). By adapting the addition formula  \cite[Eq. (22.12.8)]{abramowitz} or by direct calculation using  the generating function 
we know that for any $\tau^2$, $n$ and $m$:
\begin{equation} \label{eq:addition}
H_m^{(n \tau^2)}(x+y) = \sum_{k=0}^m \binom{m}{k} \left( \frac{n-1}{n} \right)^{k/2} \left( \frac{1}{n} \right)^{(m-k)/2} H_{m-k}^{(\tau^2)}(x) H_k^{((n-1) \tau^2)}(y).    
\end{equation}
Taking $\tau^2 = \sigma^2$ in \eqref{eq:addition}, and since for $Z$ Gaussian with mean $0$ and variance $\sigma^2(n-1)$ we know $\ep H_k^{((n-1) \sigma^2)}(Z) = 0$ for $k \geq 1$,  we can deduce that
\begin{equation} 
\Cn^* H_m^{(n \sigma^2)}(x) = \ep H_m^{(n \sigma^2)}(x+ Z) 
= \frac{1}{n^{m/2}} H_m^{(\sigma^2)}(x).
\end{equation}
Taking $\ff{f}_k = H_k^{(\sigma^2)}/\sqrt{k!}$ with 
$\fn{g}_k = H_k^{(n \sigma^2)}/\sqrt{k!}$ and $\fn{\mu}_k = n^{-k/2}$ we have $\Cn^* \fn{g}_k = \fn{\mu}_k \ff{f}_k$ as required. 

For completeness, the property that $\Cn \ff{f}_k = \fn{\mu}_k \fn{g}_k$ follows since for fixed $s$ the $Y | (S_n = s)
\sim s/n + \wt{Z}$, where $\wt{Z}$ is Gaussian with mean $0$ and variance $(n-1) \sigma^2/n$. Hence taking $\tau^2 = \sigma^2/n$ in the addition formula \eqref{eq:addition} we obtain
\begin{eqnarray*}
\Cn H_m^{(\sigma^2)}(s) = \ep H_m^{(\sigma^2)}(s/n + \wt{Z}) = \frac{1}{n^{m/2}} H_m^{(\sigma^2/n)}(s/n) = \frac{1}{n^{m/2}} H_m^{(n \sigma^2)}(s), 
\end{eqnarray*}
where the final identity follows by definition of $H_m^{(\alpha)}$.

We deduce that $\fn{\lambda}_2 = 1/n^2$  and so $\Theb{n} = n-1$, with $\Theb{2} = 1$ in particular.
\end{example}

Next, we give a similar argument in the gamma distributed case.
Note that although the $Y_i$ do not have mean 0, the argument carries through essentially unchanged on centering.

\begin{example} \label{ex:laguerre}
If $Y_i$ are   $\Gamma(\beta,1)$ distributed then, writing $L^{(\alpha)}$ for the generalized Laguerre polynomials (orthogonal with respect to $\Gamma(\alpha+1,1)$),  a similar addition formula \cite[Eq. (22.12.6)]{abramowitz} holds:
\begin{equation}
L_m^{(\alpha + \beta + 1)}(x+y) = \sum_{i=0}^m L_i^{(\alpha)}(x) L_{m-i}^{(\beta)}(y). 
\end{equation}
For $\ff{f}_k = L^{(\beta-1)}_k/\sqrt{ \binom{k + \beta-1}{k}}$ with $\fn{g}_k = L^{(\beta n - 1)}_k/\sqrt{ \binom{k + \beta n-1}{k}}$ and $\fn{\mu}_k = \sqrt{ \binom{k + \beta-1}{k}/
\binom{k + \beta n-1}{k}}$ we deduce $\Cn^* \fn{g}_k = \fn{\mu}_k \ff{f}_k$ as required.

The property that $\Cn \ff{f}_k = \fn{\mu}_k \fn{g}_k$ follows by expressing the conditional density of $Y |S_n$ in terms of a beta function and using \cite[Eq. (22.13.13)]{abramowitz}:
\begin{equation}
\Gamma( \beta n + k) \int_{0}^z (z-y)^{\beta(n-1)-1} y^{\beta-1} L_k^{(\beta-1)}(y) dy = \Gamma( \beta+k) \Gamma(\beta(n-1)) z^{\beta n - 1} L_k^{(\beta n - 1)}(z),    
\end{equation}
to deduce that $\Cn L_k^{(\beta-1)}(s) = \left( \fn{\mu}_k \right)^2 L_k^{(\beta n -1)}(s)$, and rescaling.

Hence $\fn{\lambda}_2 = \binom{\beta+1}{2}/\binom{\beta n +1}{2} = (\beta+ 1)/(n(\beta n +1))$ and so $\Theb{n} = \beta(n-1)/(\beta+1)$, with $\Theb{2} = \beta/(\beta+1)$ in particular.
\end{example}

Note that (as we might expect) the larger the value of $\beta$, the closer the value of $\Theb{2} = \beta/(\beta+1)$ obtained in Example \ref{ex:laguerre} becomes to the value $\Theb{2} = 1$ obtained for the Gaussian case in Example \ref{ex:hermite}.

Next, motivated by the fact that in both the Gaussian and gamma cases the eigenfunction $\ff{f}_2$ is quadratic, we use properties of quadratic functions to deduce an upper bound on $\Theb{n}$ involving third and fourth moments.

\begin{lemma} \label{lem:RnExample} For $Y_i$ i.i.d. $\sim Y$ with mean 0 and variance $\sigma^2$,
define the scale-invariant quantity $\Sigma = \ep Y^4/\sigma^4 - (\ep Y^3/\sigma^3)^2 - 1 \geq 0$ (kurtosis minus squared skewness minus $1$), which does not depend on $n$. Then
\begin{equation} \label{eq:RnLB}
 \Theb{n} \leq  \frac{2(n-1)}{\Sigma},\end{equation}
In particular, taking $n=2$ in \eqref{eq:RnLB} we deduce
\begin{equation} \label{eq:R2LB}
\Theb{2} \leq \frac{2}{\Sigma}.
\end{equation}
\end{lemma}
\begin{proof} Consider the function $h(s) = s^2 - a s - n \sigma^2$, where taking $a = \ep Y^3/\sigma^2$ ensures that $\ep h(S_n) S_n = \ep S_n^3 - a \ep S_n^2 = n \ep Y^3 - a n \sigma^2 = 0$ as required. Direct calculation shows that 
$\Cn^* h(y) = y^2 - a y - \sigma^2$. Further, expanding the square we can show $\ep h(S_n)^2 = n \sigma^4 \Sigma + 2 n(n-1) \sigma^4$ and $\ep \left( (\Cn^* h)(Y)^2 \right) = \sigma^4 \Sigma$ (this expression as the expectation of a square ensures that $\Sigma \geq 0$ holds).
Since it is expressed as an infimum over all functions,
$$
 \Theb{n} \leq  \frac{\ep h(S_n)^2}{n \ep 
\left( \Cn^* h \right)(Y)^2} -1  = \frac{2(n-1)}{\Sigma},$$ as required.
 \end{proof}

\begin{remark} We observe that:
\begin{enumerate}
\item Equation \eqref{eq:R2LB} shows that if $\Theb{2} > 0$ then $\Sigma < \infty$. Equivalently if $\Sigma = \infty$, we know $\Theb{2} = 0$ (and the Poincar\'{e} constant is infinite).
\item Note that  the values of $\Theb{2}$ found in Examples \ref{ex:hermite} and \ref{ex:laguerre} both satisfy \eqref{eq:R2LB} with equality, because the relevant eigenfunction $\ff{f}_2$ is quadratic. In the Gaussian case Example \ref{ex:hermite}, $\Sigma = 3 - 0 - 1 = 2$, consistent with the value $\Theta^{(2)} = 1$. In the gamma case Example \ref{ex:laguerre}, $\Sigma = (3 + 6/\beta) - 4/\beta - 1 = 2 + 2/\beta$, consistent with the value $\Theb{2} = \beta/(\beta + 1)$.
\item 
Note also that \eqref{eq:R2LB} means that if $\Sigma > 2$ (which, roughly speaking, corresponds to $Y$ having heavier tails than the Gaussian) then by \eqref{eq:RnLB} the $\Theb{2} < 1$ (smaller than the value in the Gaussian case, Example \ref{ex:hermite}).
\end{enumerate}
\end{remark}

Indeed, we can prove similar (if more involved) bounds which show that positivity of $\Theb{2}$ implies finiteness of all moments.  As before,the following proposition implies that if 
$\Theb{2} > 0$ and the $(2k-2)$th moment of $Y$ is finite then the $(2k)$th moment of $Y$ must be finite.

\begin{proposition} \label{prop:R2bound}
Writing $m_k = \ep Y^k$ for the $k$th moment of $Y$ and $\sigma^2$ for its variance, there exist functions $B_{1,k}$ and $B_{2,k}$ (depending on moments of lower orders) such that
$$ \Theb{2} \leq \frac{\sigma^2 m_{2k-2} + B_{1,k}(m_1, \ldots, m_{2k-1}) }{2 m_{2k}  + B_{2,k}(m_1, \ldots, m_{2k-2})}.$$ 
\end{proposition}
\begin{proof} Write
\begin{equation} \label{eq:mk2}
M_k = \ep(Y_1 + Y_2)^k = \sum_{l=0}^k \binom{k}{l} m_l m_{k-l} \end{equation} for the $k$th moment of $S_2$.
As in Lemma \ref{lem:RnExample}, consider the function $h(s) = s^k - a s - M_k$, where taking $a= M_{k+1}/(2 \sigma^2)$ ensures that $\ep h(S_2) S_2 = 0$.
Using \eqref{eq:mk2} we can expand
\begin{eqnarray} 
h(y_1 + y_2) & = &  (y_1 + y_2)^k  - a (y_1 + y_2) - M_k \nonumber \\
& = & (y_1^k - a y_1 - m_k) +   (y_2^k - a y_2 - m_k) + \sum_{l=1}^{k-1} \binom{k}{l} \left( y_1^l y_2^{k-l} - m_l m_{k-l} \right). \label{eq:hsumexp}
\end{eqnarray}
Substituting $y_2 = Y_2$ and taking expectations, we deduce that
\begin{eqnarray}
\Cn^* h(y) & = & \ep h(y + Y_2) =  (y^k - a y - m_k) +  \sum_{l=1}^{k-1} \binom{k}{l} \left( y^l  - m_l \right) m_{k-l},  
\end{eqnarray}
meaning that we can rewrite \eqref{eq:hsumexp} as
\begin{eqnarray*}
\lefteqn{ h(y_1 + y_2) } \\
 & = & \left( \Cn^* h(y_1) -  \sum_{l=1}^{k-1} \binom{k}{l} \left( y_1^l  - m_l \right) m_{k-l} \right)  + \left( C^* h(y_2) -  \sum_{l=1}^{k-1} \binom{k}{l} \left( y_2^l  - m_l \right) m_{k-l} \right)  \\
& & + \sum_{l=1}^{k-1} \binom{k}{l} \left( y_1^l y_2^{k-l} - m_l m_{k-l} \right) \\
& = & \Cn^* h(y_1) + \Cn^* h(y_2)  + U(y_1,y_2)
\end{eqnarray*} 
where $$ U(y_1,y_2) := \sum_{l=1}^{k-1} \binom{k}{l} \left( y_1^l - m_l \right) \left( y_2^{k-l} - m_{k-l} \right).$$
Since by construction $\ep \Cn^*h(Y)= 0$, we can deduce by independence of $Y_1$ and $Y_2$ that the cross terms vanish so that
\begin{equation} \label{eq:functionexpand} \ep h(Y_1+Y_2)^2 = 2 \ep \left( \Cn^* h(Y) \right)^2 + \ep U(Y_1,Y_2)^2,\end{equation}
so that 
$$ \Theb{2} \leq \frac{1}{2} \frac{ \ep h(Y_1+Y_2)^2}{  \ep \left( \Cn^* h(Y) \right)^2 }  -1 = \frac{ \ep U(Y_1,Y_2)^2}{ 2 \ep \left( \Cn^* h(Y) \right)^2 }.$$
We deduce the result using \eqref{eq:functionexpand} and the facts that $ \ep h(Y_1 + Y_2)^2 = M_{2k} - M_k^2 - M_{k+1}^2/\sigma^2$ and 
$$ \ep U(Y_1,Y_2)^2 =  \sum_{l,j = 1}^{k-1} \binom{k}{l} \binom{k}{j} \left( m_{l+j} - m_l m_j \right) \left( m_{2k-l-j} - m_{k-l} m_{k-j} \right).$$
\end{proof}

\begin{lemma} \label{lem:weaker}
Assuming $J(Y) < \infty$, the finiteness of the Poincar\'{e} constant $C_P := C_P(Y)$ implies that $\Theb{2} > 0$. Indeed: 
\begin{equation} \Theb{2} \geq \frac{1}{2 J(Y) C_P},
\end{equation}
\end{lemma}
\begin{proof}
We can deduce this using  \cite[Proposition 2.1]{johnson5} which, for $Y_1$ and $Y_2$ i.i.d., gives that for any $f$ with $\ep f(S_2) = \ep f(S_2) S_2 = 0$ and
taking $g(u) = \ep f(u+ Y)$ then
\begin{equation} \label{eq:johnsonbarron}
 \ep \left( f(Y_1 + Y_2) - g(Y_1) - g(Y_2) \right)^2
\geq \frac{1}{J(Y) C_P} \ep \left( g(Y_1) - \mu Y_1 - \nu \right)^2
\end{equation}
for some $\mu$, $\nu$.
The proof of \cite[Proposition 2.1]{johnson5} states that $\nu = \mu \ep Y_1  = 0$. Further, by symmetry, the condition $\ep f(S_2) S_2 = 0$ implies that $0 = \ep f(Y_1 +Y_2) Y_1 = \ep g(Y_1) Y_1$, so the RHS of \eqref{eq:johnsonbarron} is $ \geq \frac{1}{J(Y) C_P} \ep \left( g(Y) \right)^2$. Rearranging, we deduce that
\begin{equation}
\ep f(S_2)^2 \geq \left( 2+ \frac{1}{J(Y) C_P} \right) \ep g(Y)^2, \end{equation}
and the result follows on rearranging.
\end{proof}

\section{Behaviour of the Fisher information on convolution} \label{sec:scoreconv}

We now consider how the standardized Fisher information behaves on convolution, under a standard Central Limit Theorem scaling.
That is, as in \cite{courtade2016}, we write $U_n = S_n/\sqrt{n}$.
Note that in the i.i.d. regime, since $J(c X) = J(X)/c^2$ (see e.g. \cite[Eq. (2.3)]{brown}) we know that $\Jst(U_n)  = \sigma^2 J(U_n) - 1 = \sigma^2 J(S_n/\sqrt{n}) - 1 = \sigma^2 n J(S_n) - 1$ (scale-invariance of $\Jst$).

\begin{proposition} \label{prop:fishbd1}
For i.i.d. $Y_1, \ldots, Y_n$ the standardized Fisher information satisfies
\begin{equation} \Jst(U_n) \leq \frac{1}{1 + \Theb{n}}\Jst(Y).\label{eq:fishUB} \end{equation}
\end{proposition}
\begin{proof}
Observe that (see for example \cite[Eq. (3)]{courtade2016}, \cite{stam}) that the score function of the sum satisfies 
\begin{equation} \label{eq:scoresum} \varrho_{S_n}(s) = \ep \left[ \varrho_{Y}(Y_i) | S_n = s \right],\end{equation}
which we can rewrite as  $\varrho_{S_n} = (\Cn \varrho_Y)$.
Hence if we expand the score function as a sum of eigenfunctions
$$ \varrho_Y = \sum_{k=1}^\infty a_k \ff{f}_k$$ then
Definition \ref{def:eigenfunctions} gives that:
\begin{equation} \label{eq:scoresumexp} \varrho_{S_n} = \sum_{k=1}^\infty a_k (\Cn \ff{f}_k)
= \sum_{k=1}^\infty a_k \fn{\mu}_k \fn{g}_k. \end{equation}

Further, direct calculation using integration by parts gives that 
$$ a_1 = \left\langle \varrho_Y, \ff{f}_1 \right\rangle_{\pr_Y} = \int p_{Y}(y) \frac{ p'_{Y}(y)}{p_{Y}(y)} \frac{y}{\sigma} dy
= - \int p_{Y}(y) \frac{1}{\sigma} = - \frac{1}{\sigma}.$$
This means that, using the fact that (see Remark \ref{rem:eigen}.\ref{it:eigen1}) the $\ff{f}_1(x) = x/\sigma$, $\fn{g}_1(y) = y/(\sigma \sqrt{n})$ with $\fn{\mu}_1 = 1/\sqrt{n}$
we can write the standardized score functions of $Y$ and $S_n$ from
\eqref{eq:jstmain} as sums of eigenfunctions starting at index 2, as:
\begin{eqnarray}
\varrho_Y(y) + \frac{y}{\sigma^2} & = & \sum_{k=2}^\infty a_k \ff{f}_k(y), \\
\varrho_{S_n}(s) + \frac{s}{n \sigma^2} 
& = & \sum_{k=2}^\infty a_k \fn{\mu}_k \fn{g}_k(s).
\end{eqnarray}
Then, direct calculation using the orthonormality of $\ff{f}$ and $\fn{g}$ gives that:
\begin{eqnarray} 
  \Jst(U_n)  =  \Jst(S_n) 
& = & (\sigma^2  n) \ep \left( \varrho_{S_n}(S_n) + \frac{S_n}{n \sigma^2} \right)^2 \nonumber \\
& = & (\sigma^2  n) \left( \sum_{k=2}^\infty a_k^2 \left( \fn{\mu}_k \right)^2 \right)   \nonumber \\
& \leq & \frac{1}{1+\Theb{n}} \sigma^2 \left( \sum_{k=2}^\infty a_k^2 \right) \nonumber \\
& = & \frac{1}{1+ \Theb{n}} \sigma^2 \ep \left( \varrho_Y(Y) + \frac{Y}{\sigma^2} \right)^2 \nonumber \\
& = &  \frac{1}{1+\Theb{n}} \Jst(Y),
\label{eq:fisherdrop2}
\end{eqnarray}
using the fact that $n \left( \fn{\mu}_k \right)^2 = n \fn{\lambda}_k \leq  1/(1+\Theb{n})$ for $k \geq 2$ by \eqref{eq:higher}.
\end{proof}

We can use a similar argument to prove the lower bound on Fisher information, Lemma \ref{lem:fishbd3}:

\begin{proof}[Proof of Lemma \ref{lem:fishbd3}]
As in Lemma \ref{lem:RnExample} consider the function $h(s) = s^2 - a s - n \sigma^2$ where $a= \ep Y^3/\sigma^2 = \skewn \sigma$. As above, since $\ep h(S_n) S_n = 0$
\begin{equation} \label{eq:findinnprod}
\left\langle \varrho_{S_n} + \frac{S_n}{n \sigma^2}, h \right\rangle_{\pr_{S_n}}
= \left\langle \varrho_{S_n}, h \right\rangle_{\pr_{S_n}}
= \int p_{S_n}(y) \frac{ p'_{S_n}(y)}{p_{S_n}(y)} h(y) dy
= - \int p_{S_n}(y) h'(y) = \skewn \sigma.\end{equation}
Now considering the LHS of \eqref{eq:findinnprod} using Cauchy-Schwarz, we deduce that
\begin{equation}
\left( \skewn \sigma \right)^2 \leq \left( \ep \left( \varrho_{S_n} + \frac{S_n}{n \sigma^2} \right)^2 \right)   \left( \ep h(S_n)^2 \right)^2 = \frac{\Jst(S_n)}{n \sigma^2}  \left( n \sigma^4 \Sigma + 2 n(n-1) \sigma^4 \right).
\end{equation}
since as before $\ep h(S_n)^2 = n \sigma^4 \Sigma + 2 n(n-1) \sigma^4$, and the result follows by rearrangement.
\end{proof}

This lower bound tightens \cite[Lemma 1.4]{johnson5}, which (in our notation) can be expressed as $$ \Jst(U_n) \geq \frac{\skewn^2}{\Sigma + \skewn^2 + 2(n-1)}, $$ where the original result is expressed in terms of the excess kurtosis $k = \ep Y^4/\sigma^4 -3  = \Sigma + \skewn^2 - 2$.

\section{Higher order Dembo--Kagan--Shepp terms} \label{sec:DKShigher}

Proposition \ref{prop:fishbd1} gives one part of the proof of Theorem \ref{thm:main}. However, this result as stated is not particularly helpful, since the form of the dependence of $\Theb{n}$ on $n$ is not immediately clear. We complete the proof of Theorem \ref{thm:main} by proving Proposition \ref{prop:DKShigher} below, which allows us to control $\Theb{n}$.

The key observation is that we can analyse higher order terms in the Dembo--Kagan--Shepp argument, following the proof of \cite[Lemma 2]{DKS2001}. 

\begin{lemma} \label{lem:DKShigher}
Fix $k > \ell \geq 2$, and consider a function $h$ with $\ep h(S_k) = 0$. Then
$$ \ep h(S_k)^2 \geq k \ep h_1(Y)^2 + \frac{k(k-1)}{\ell(\ell-1)}
\left( \ep \widehat{h}(S_\ell)^2 - \ell \ep h_1(Y)^2 \right),$$
where $h_1(u) = \ep h( u + Y_2 + \ldots + Y_k)$ and $\widehat{h}(v) 
= \ep  h(v + Y_{\ell+1} + \ldots + Y_k)$.
\end{lemma}
\begin{proof}
We adopt the same notation as \cite[Section 2]{DKS2001}.
As in \cite[Eq. (14), (15)]{DKS2001}, we can perform an Efron--Stein (ANOVA) expansion \cite{efronstein} of $h$ and $\widehat{h}$ (using the same functions $h_i$ in each case) to obtain
\begin{eqnarray}
\ep h(S_k)^2 & = & \sum_{r=1}^k \binom{k}{r} \ep h_r^2(Y_1, \ldots, Y_r)
\label{eq:ES1} \\
\ep \widehat{h}(S_\ell)^2 & = & \sum_{r=1}^\ell \binom{\ell}{r} \ep h_r^2(Y_1, \ldots, Y_r)
\label{eq:ES2}
\end{eqnarray}
 The key observation is that for any $k > \ell \geq 2$ and any $r \geq 2$, direct comparison of the two terms gives
$$ \frac{ \ell(\ell-1)}{k(k-1)} \binom{k}{r} \geq \binom{\ell}{r},$$
with equality if and only if $r = 2$. Applying this to the Efron--Stein decompositions \eqref{eq:ES1} and \eqref{eq:ES2} we obtain
\begin{eqnarray}
\ep h(S_k)^2 & = & k \ep h_1(Y)^2
+\sum_{r=2}^k \binom{k}{r} \ep h_r^2(S_r) \nonumber \\
& \geq & k \ep h_1(Y)^2
+ \frac{k(k-1)}{\ell(\ell-1)} \left( \sum_{r=2}^\ell \binom{\ell}{r} \ep h_r^2(S_r) \right) \nonumber \\
& \geq & k \ep h_1(Y)^2 + \frac{k(k-1)}{\ell(\ell-1)}
\left( \ep \widehat{h}(S_\ell)^2 - \ell \ep h_1(Y)^2 \right)
 \label{eq:DKS2}
\end{eqnarray}
as required.
\end{proof}

We now deduce a result which, when combined with Proposition \ref{prop:fishbd1} above, allows us to deduce the proof of Theorem \ref{thm:main}:

\begin{proposition} \label{prop:DKShigher}
The quantity $\Theb{k}/(k-1)$ is non-decreasing in $k$. Specifically for any $n \geq 2$:
\begin{equation} \label{eq:LamBD}
 \Theb{n} \geq (n-1) \Theb{2}.\end{equation}
\end{proposition}
\begin{proof}
The key fact is that the function $h_1$ arising in Lemma \ref{lem:DKShigher} can be understood as the conditional expectation of both $h$ and $\widehat{h}$ (this is remarked at the foot of \cite[P.345]{DKS2001}, and is due to orthogonality of the Efron--Stein decomposition). That is, for $k > \ell \geq 2$ we can write
\begin{equation}
h_1(u)  =   \ep h(u+ Y_2 + \ldots + Y_k) = \left\{ 
\mbox{
\begin{tabular}{llll}
$C^{(k)*} h(u)$ \\
$\ep \widehat{h}(u + Y_2+ \ldots + Y_\ell)$ & $= C^{(\ell)*} \widehat{h}(u)$ \\
\end{tabular}}
\right.
\end{equation}
since $\widehat{h}(v) = \ep h(v+ Y_{\ell+1} + \ldots + Y_k)$. Hence for any $h$ (and hence $\widehat{h}$ and $h_1$) we can write
\begin{equation} \label{eq:variation}
\Theb{\ell} \leq  \frac{ \ep \widehat{h}(S_\ell)^2}{\ell \ep h_1(Y_1)^2} - 1,
\end{equation}
 so the RHS of \eqref{eq:DKS2} becomes 
$$ \ep h(S_k)^2 \geq
k \ep h_1(Y)^2 \left( 1 + \frac{(k-1)}{(\ell-1)} \Theb{\ell}  \right), $$
or dividing by $k \ep h_1(Y)^2$ and taking the optimal $h$:
$$ \Theb{k} \geq \frac{k-1}{\ell-1} \Theb{\ell},$$
and the result \eqref{eq:LamBD} follows on taking $k = n$ and $\ell = 2$. \end{proof}

Note that we can weaken the assumption that $\Theb{2} > 0$ to ensure $O(1/n)$ convergence of Fisher information, to simply require that $\Theb{m} > 0$ for some $m$. If this is true, we can simply replace \eqref{eq:LamBD} by a bound of the form $\Theb{n} \geq (n-1) \Theb{m}/(m-1)$ and substitute this in Proposition \ref{prop:fishbd1} instead.

\begin{example}
In the Gaussian and gamma cases of Examples \ref{ex:hermite} and \ref{ex:laguerre} the result of Proposition \ref{prop:DKShigher} is sharp. That is, if $Y_i \sim N(0,\sigma^2)$ then recall that $\Theb{k} = k-1$ and hence $\Theb{k}/(k-1)  \equiv 1$. Similarly if $Y_i \sim \Gamma(\beta,1)$ then $\Theb{k} = (k-1) \beta/(\beta+1)$ and $\Theb{k}/(k-1) \equiv \beta/(\beta+1)$.

This sharpness holds because in both Example \ref{ex:hermite} and \ref{ex:laguerre} the optimal eigenfunction is quadratic, so in the Efron--Stein decomposition the $h_3 = h_4 = \ldots = 0$.
\end{example}

\begin{remark} By combining Proposition \ref{prop:DKShigher} with Equation \eqref{eq:RnLB}  we can deduce that $\Theb{n}$ is bounded above and below by linear functions in $(n-1)$, assuming $\Theb{2}$ and $\Sigma$ are non-zero, as
$$  (n-1) \Theb{2} \leq \Theb{n}
\leq  (n-1) \frac{2}{\Sigma}.$$
\end{remark}

\begin{remark} \label{rem:unique}
Although not mentioned in \cite{DKS2001}, similar arguments show that under regularity conditions there should be a unique eigenfunction achieving eigenvalue $1/n$ (we know from Example \ref{rem:eigen}.\ref{it:iid} above that the linear functions achieve this). That is, assuming $k > \ell$ there is equality in
$$ \frac{\ell}{k} \binom{k}{r} \geq \binom{\ell}{r}$$
if and only if $r = 1$. Hence there is equality in \cite[Lemma 2]{DKS2001} if and only if $\ep h_2^2 = \ep h_3^2 = \ldots = 0$. Hence except on a set of measure 0 we know that
$$ h( y_1 + \ldots + y_k) = \sum_{i=1}^k h_1(y_i).$$
Assuming $h$ is twice differentiable then taking a derivative with respect to $y_1$ and $y_2$ this implies $h''(z) \equiv 0$ for all $z$, so $h$ is linear.
\end{remark}

\section{Strengthened monotonicity} \label{sec:strengthen}

We can extend the arguments above to deduce a stronger form of
monotonicity of Fisher information than that obtained by \cite{artstein} and \cite{courtade2016}, at least in the i.i.d. case:

\begin{definition} \label{def:higher2} For $m \leq n$, define $\CnG{n}{m}^*$ by
$\left( \CnG{n}{m}^* g \right)(y) = \ep \left[ g(S_n) | S_m = y \right]$, and write $\fnG{\lambda}{n}{m}_k$ for the ordered eigenvalues of $\CnG{n}{m}^* \CnG{n}{m}$, where $\fnG{\lambda}{n}{m}_0 = 1$ and (by DKS \cite{DKS2001} \eqref{eq:DKS}) $\fnG{\lambda}{n}{m}_1 = m/n$. Again, write $\fnG{\mu}{n}{m}_k = \sqrt{ \fnG{\lambda}{n}{m}_k}$.

Define a generalization of $\Theb{n}$ as
$$ \TheG{n}{m} := \frac{m}{n  \fnG{\lambda}{n}{m}_2   } -1 =
\inf_{h: \ep h(S_n) = \ep S_n h(S_n) = 0} \frac{m \ep h(S_n)^2}{n \ep \left( (\CnG{n}{m}^* h) (S_m)^2 \right)}  -1 . $$
\end{definition}

As before, the Dembo--Kagan--Shepp identity \cite{DKS2001} ensures that $\TheG{n}{m} \geq 0$.
Note we recover Definition \ref{def:higher} by taking $m=1$. We now give a result which generalizes Proposition \ref{prop:fishbd1}.

\begin{proposition} \label{prop:fishbd2}
For i.i.d. $Y_1, \ldots, Y_n$ the standardized Fisher information satisfies
$$ \Jst(U_n) \leq \frac{1}{1+ \TheG{n}{m}} \Jst(U_m), 
\mbox{ \;\;\; for $m \leq n$.}$$
\end{proposition}
\begin{proof} We repeat the steps of the proof of Proposition \ref{prop:fishbd1}.
Again (see for example \cite[Eq. (3)]{courtade2016}, \cite{stam}) the score function of the sum satisfies 
$$ \varrho_{S_n}(s) = \ep \left[ \varrho_{S_m}(S_m) | S_n = s \right],$$
which we can rewrite as  $\varrho_{S_n} = (\CnG{n}{m} \varrho_{S_m})$.
Hence if we expand the score function 
$$ \varrho_{S_m} = \sum_{k=1}^\infty a_k \fm{f}_k$$ then
\begin{equation} \label{eq:scoresumexp2} \varrho_{S_n} = \sum_{k=1}^\infty a_k (\fnG{C}{n}{m} \fm{f}_k)
= \sum_{k=1}^\infty a_k \fnG{\mu}{n}{m}_k \fn{g}_k. \end{equation}
As before, direct calculation using integration by parts gives that 
$$ a_1 = \left\langle \varrho_{S_m}, \fm{f}_1 \right\rangle_{\pr_{S_m}} = \int p_{S_m}(y) \frac{ p'_{S_m}(y)}{p_{S_m}(y)} \frac{y}{\sigma \sqrt{m}} dy = - \int p_{S_m}(y) \frac{1}{\sigma \sqrt{m}} = - \frac{1}{\sigma \sqrt{m}}.$$
Again (as in Remark \ref{rem:eigen}.\ref{it:eigen1}) the $\fm{f}_1(y) = y/(\sigma \sqrt{m})$, $\fn{g}_1(s) = s/(\sigma \sqrt{n})$ with $\fnG{\mu}{n}{m}_1 = \sqrt{m/n}$ so
we can write the standardized score functions of $S_m$ and $S_n$ from
\eqref{eq:jstmain} as sums of eigenfunctions starting at index 2, as:
\begin{eqnarray}
\varrho_{S_m}(y) + \frac{y}{m \sigma^2} & = & \sum_{k=2}^\infty a_k \fm{f}_k(y), \\
\varrho_{S_n}(s) + \frac{s}{n \sigma^2} 
& = & \sum_{k=2}^\infty a_k \fnG{\mu}{n}{m}_k \fn{g}_k(s).
\end{eqnarray}
Just as before, we can use the orthonormality of $\fm{f}$ and $\fn{g}$ to deduce
\begin{eqnarray} 
  \Jst(U_n) = \Jst(S_n) & = & (n \sigma^2) \ep \left( \varrho_{S_n}(S_n) + \frac{S_n}{n \sigma^2} \right)^2 \\
& = & \sigma^2  n \left( \sum_{k=2}^\infty a_k^2 \fnG{\lambda}{n}{m}_k \right)   \nonumber \\
& \leq & \frac{1}{1+\TheG{n}{m}} m \sigma^2 \left( \sum_{k=2}^\infty a_k^2 \right) \nonumber \\
& = & \frac{1}{1+\TheG{n}{m}}  \Jst(S_m)
\label{eq:fisherdrop3}
\end{eqnarray}
using the fact that  $n \fnG{\lambda}{n}{m}_k  \leq m/(1+\Theb{n,m})$ for $k \geq 2$.
\end{proof}

As in \cite{courtade2016}, taking the Dembo--Kagan--Shepp bound $\TheG{n}{m} \geq 0$ in Proposition \ref{prop:fishbd2} we recover the monotonicity of standardized Fisher information \cite{artstein}. However, we can obtain better bounds by taking $k = n$ and $\ell = m$ in Lemma \ref{lem:DKShigher} to obtain
\begin{eqnarray*}
 \ep h(S_n)^2 & \geq & n \ep h_1(Y)^2 + \frac{n(n-1)}{m(m-1)}
\left( \ep \widehat{h}(S_m)^2 - m \ep h_1(Y)^2 \right) \\
& = &  \frac{n(n-1)}{m(m-1)}
\ep \widehat{h}(S_m)^2 - \frac{n (n-m)}{m-1} \ep h_1(Y)^2 \\
& \geq & \frac{n}{m(m-1)} \left( n-1 - \frac{(n-m)}{1+\Theb{m}}
\right) \ep \widehat{h}(S_m)^2.
\end{eqnarray*}
Rearranging, and optimizing over $h$ we deduce that \begin{equation} \label{eq:LamBD2}
1 + \Theb{n,m} = \frac{m \ep h(S_n)^2}{n\ep \widehat{h}(S_m)^2 }  \geq
\frac{1}{(m-1)} \left( n-1 - \frac{(n-m)}{1+\Theb{m}} \right)
 = \frac{ (n-1) \Theb{m} + (m-1)}{(m-1)(1+ \Theb{m})}.
\end{equation}
Since this is an increasing function of $\Theb{m}$, we can replace
$\Theb{m}$ by the lower bound $(m-1) \Theb{2}$ from \eqref{eq:LamBD}
to obtain 
$$ 1 + \Theb{n,m} \geq \frac{ 1+  (n-1) \Theb{2}}{ 1+  (m-1) \Theb{2}},$$
which, in Proposition \ref{prop:fishbd2} allows us to deduce the
 stronger form of monotonicity that:
 \begin{theorem}  
 \label{thm:main2}
 Consider i.i.d. $Y_i \sim Y$ with mean $0$ and variance $\sigma^2 < \infty$ and smooth densities on $\re$. Writing $\Theb{2}$ for the quantity from Definition \ref{def:higher},
 the standardized Fisher information has the property that
\begin{equation} \label{eq:jstmain2}
\left( 1 + (n-1) \Theb{2} \right) \Jst \left( \frac{Y_1 + \ldots + Y_n}{\sqrt{n}} \right) \mbox{ \;\;\; is non-increasing in $n$.}
 \end{equation}
 \end{theorem}
 Note that this is a simultaneous strengthening of Theorem \ref{thm:main} and of the monotonicity of Fisher information proved in the i.i.d. case by Artstein et al. \cite{artstein} and \cite{courtade2016}.

\section{Future work} \label{sec:futurework}

We briefly mention some future directions for research. Note that some progress is made towards 1. and 2. in Appendix \ref{sec:subgausspf} below:
\begin{enumerate}
    \item In order to increase the value of these results, it is a natural question to ask for sufficient conditions (in terms of the density $p_Y$ or other related quantities) under which $\Theb{2} > 0$, and indeed to give explicit bounds of the form $\Theb{2} \geq c$ for some $c > 0$.
    \item Additionally, it would be of value to give conditions on $U$ under which we can bound $\Theb{2}$ uniformly away from $0$ for all $t> 0$, for random variables of the form $Y = U + Z_t$, where $Z_t$ is an independent Gaussian perturbation. Such a result would allow us to derive $O(1/n)$ convergence of relative entropy using the de Bruijn identity \cite{stam}.
    \item Since the monotonicity of entropy is equivalent to strengthened forms of Shannon's Entropy Power Inequality (see \cite{artstein, madiman, tulino}), it would be of interest to know if the strengthened monotonicity result Theorem \ref{thm:main2} implies a stronger Entropy Power Inequality.
    \item The results of this paper very much rely on the i.i.d. assumption. It is of interest to weaken this to the independent, but not identical setting, and indeed to dependent random variables, for example in the exchangeable setting. For example, Peccati \cite{peccati} shows that a decomposition of the Efron--Stein type used to establish the Dembo--Kagan--Shepp identity holds if an exchangeable sequence has the `weak independence' property. It is a natural question whether the results of this paper hold in that setting.
\item Following recent trends in information-theoretic Central Limit Theorems, it would be of interest to extend the results of this paper to the setting of $\re^d$, and to understand the behaviour of the eigenfunctions of $\Cn^* \Cn$ in this setting, where the equivalent of \eqref{eq:scoresum} still holds
(see e.g. \cite[Lemma 3.4]{johnson14}).
\end{enumerate}

\appendix

\section{Proof of Theorem \ref{thm:subgauss}} \label{sec:subgausspf}

The following argument was provided by an anonymous referee, for which the author is extremely grateful.

We write $\gamma_{\vd{\mu}; \Sigma}$ for a Gaussian density centred at $\vd{\mu}$ with covariance matrix $\Sigma$. As before, we write $D_{\chi^2}(f \| g) = \int (f(x)/g(x) -1)^2 g(x) dx= \int f(x)^2/g(x) dx - 1$ for the $\chi^2$-divergence. We first state two lemmas:

\begin{lemma} \label{lem:fdiv} For any coupling of $\vc{V} \sim p$ and $\vc{W} \sim q$:
$$ D_{\chi^2} \left( p \star \gamma_{\vc{0}; \Sigma_1} \| q \star \gamma_{\vc{0}; \Sigma_2} \right) \leq \ep \left( D_{\chi^2}( \gamma_{\vc{V}; \Sigma_1} \|  \gamma_{\vc{W}; \Sigma_2})
\right). $$
\end{lemma} 
\begin{proof} Follows immediately from the joint convexity of $f$-divergences (see for example \cite[Lemma 4.1]{csiszar2004information}). \end{proof}

\begin{lemma} \label{lem:chi2val} For any $\rho \in (-1,1)$, write $I_2$ for the two dimensional identiy matrix, and define the positive semi-definite matrix 
$$ R_\rho = \left( \begin{array}{cc} 1 & \rho \\ 
\rho & 1 \\ \end{array} \right),$$
then for any $\vc{x}, \vc{y} \in \re^2$ and $\delta > 0$:
\begin{eqnarray*}
D_{\chi^2} \left( \gamma_{\vc{x}; \delta^2 R_\rho} \| \gamma_{\vc{y}; \delta^2 I_2} \right)
& = & \frac{1}{1-\rho^2} \exp \left( \frac{ (\vc{x} - \vc{y})^T R_\rho  (\vc{x} - \vc{y})}{(1-\rho^2) \delta^2} \right) - 1.
\end{eqnarray*}
\end{lemma}
\begin{proof}
The key is that for $\vc{u} \in \re^2$ we can express the ratio
\begin{eqnarray*}
\frac{ \gamma_{\vc{x}; \delta^2 R_\rho}^2(\vc{u})}{ \gamma_{\vc{y}; \delta^2 I_2}(\vc{u})}
\left( \frac{1}{1-\rho^2} \exp \left( \frac{ (\vc{x} - \vc{y})^T R_\rho  (\vc{x} - \vc{y})}{(1-\rho^2) \delta^2} \right) \right)^{-1}
\end{eqnarray*}
as a product of Gaussian densities, which integrate to 1.
\end{proof}

\begin{proof}[Proof of Theorem \ref{thm:subgauss}]
Write $X_i$ for i.i.d. copies of $X \sim p_X$ and independent $Z_i \sim \gamma_{0;\delta^2}$, and define regularized $Y_i = X_i + Z_i$. Define $V_n = n^{-1/2} \sum_{i=1}^n  X_i $ and
$U_n = n^{-1/2} \sum_{i=1}^n  Y_i $. Further, define $X_i'$ and $Z_i'$ to be independent copies of $X_i$ and $Z_i$ respectively, write $Y_i' = X_i' + Z_i'$, and define $V_n' = n^{-1/2} (  X_1' + X_2 + \ldots + X_n)$ and $U_n' = n^{-1/2} (  Y_1' + Y_2 + \ldots + Y_n)$.

Using the invariance of $f$-divergences under $1-1$ mappings we can write
$$ D_{\chi^2}( p_{Y_1, S_n} \| p_{Y} \times p_{S_n}) =   D_{\chi^2}( p_{Y_1, U_n} \| p_{Y} \times p_{U_n}).$$
For any $n$, we can consider the coupling between $p_{X,V_n}$ and $p_X \times p_{V_n}$ given by $\left( (X_1,V_n), (X_1, V_n') \right)$, and 
note that $(X_1 , V_n) - (X_1,V_n') = (0, n^{-1/2} (X_1 - X_1')$. This means that we can express the $\chi^2$-divergence 
arising in the formula for $T_n(Y)-1$ as
\begin{eqnarray}
 D_{\chi^2} ( p_{Y_1, U_n} \| p_{Y} \times p_{U_n}) &  \leq & \ep \left( D_{\chi^2} ( \gamma_{(X_1, V_n); \delta^2 R_\rho} \| \gamma_{ (X_1,V_n'); \delta^2 I_2}
\right) \nonumber \\
& = & \frac{1}{1-1/n} \ep \left( \exp \left( \frac{ (X_1 - X_1')^2}{(n-1) \delta^2} \right) \right) \label{eq:chi2bound2}
\end{eqnarray}
where we apply Lemma \ref{lem:fdiv} followed by Lemma \ref{lem:chi2val}, and where $\rho = 1/\sqrt{n}$.

If $X$ is sub-Gaussian, then (see for example \cite[Proposition 2.6.1]{vershynin}) so is $(X-X')$, and hence (see \cite[Proposition 2.5.2]{vershynin}) there exists a constant $\theta$ such that the moment generating function of $(X-X')^2$ satisfies
$$ \ep \exp \left( \lambda^2 (X-X')^2 \right) \leq \exp( \theta^2 \lambda^2),$$
when $|\lambda| \leq 1/\theta$. Hence taking $\lambda^2 = \frac{1}{(n-1) \delta^2}$ we deduce that \eqref{eq:chi2bound2} is bounded above by
$$ \frac{n}{n-1}  \exp \left( \frac{\theta^2}{(n-1) \delta^2} \right),$$
assuming that $\theta^2 \leq 1/\lambda^2$, or equivalently $n \geq 1 + \theta^2/\delta^2$.
\end{proof}

Observe that we can use \eqref{eq:chi2bound2} to deduce asymptotic bounds on $\Theb{n}$ for sub-Gaussian random variables under sufficiently large amounts of Gaussian regularization.
That is, if we write $\lambda_k^{n,\delta}$ or the eigenvalues of the operator $\Cn^* \Cn$, then since we can express the trace as
\begin{eqnarray*}
1 + \frac{1}{n}  + \sum_{k=2}^\infty \lambda_k^{n,\delta} & = & 1 +  D_{\chi^2} ( p_{Y_1, U_n} \| p_{Y} \times p_{U_n}) \leq 1 + \frac{n}{n-1} \ep \left( \exp \left( \frac{ (X_1 - X_1')^2}{(n-1) \delta^2} \right) \right),
\end{eqnarray*} 
so that Taylor expanding the exponential
\begin{eqnarray*}
n \sum_{k=2}^\infty \lambda_k^{n,\delta}
& \leq & \frac{1}{n-1} + \frac{2 n^2}{n-1} \frac{\var(X)}{(n-1) \delta^2} + \frac{n^2}{n-1} \sum_{r=2}^\infty \frac{ \ep(X_1-X_1')^{2r}}{r! (n-1)^r \delta^{2r}} \\
& \leq & \frac{1}{n-1} + \frac{2 n^2}{(n-1)^2 \delta^2} \var(X) + \frac{n^2}{n-1} \sum_{r=2}^\infty \frac{ \phi^{2r}}{(n-1)^r \delta^{2r}}
\end{eqnarray*}
using the fact that  (see \cite[P.26]{vershynin}) $\ep |X_1 - X_1'|^{2r} \leq \phi^{2r} (2r) \Gamma(r) = 2 \phi^{2r} r!$. Hence we deduce:

\begin{corollary} Writing $Y_i = X_i + Z_i$ where $X_i$ is sub-Gaussian and $Z_i \sim \gamma_{0;\delta^2}$, and writing $\lambda_k^{n,\delta}$ for the eigenvalues of the operator $\Cn^* \Cn$ we deduce:
\begin{equation} \label{eq:unifbd}  \limsup_{n \rightarrow \infty} n \lambda_2^{n,\delta} \leq
\limsup_{n \rightarrow \infty} n \sum_{k=2}^\infty \lambda_k^{n,\delta} \leq \frac{2\var(X)}{\delta^2},\end{equation}
and hence $\Theb{n} > 0$ for $n$ sufficiently large for sub-Gaussian random variables regularized by a sufficiently large amount.
\end{corollary}

Note that this allows us to deduce $O(1/n)$ convergence of Fisher information for random variables of this type. Indeed, it allows us to deduce $O(1/n)$ convergence of relative entropy using the de Bruijn identity (see for example \cite[Eq. (1.110]{johnson14}) which expresses the relative entropy of a random variable $U$ with density $f$ to a standard Gaussian density $\gamma_{0,1}$ as the integral of standardized Fisher information
\begin{equation} \label{eq:debruijn}
 D( f \| \gamma_{0,1}) = \frac{\log e}{2} \int_0^\infty \frac{ \Jst(U + Z_\tau)}{1+\tau} d\tau. \end{equation}
Using \eqref{eq:unifbd}, since adding $Z_\tau$ provides extra regularization we can deduce bounds for all $\tau$ on the second--largest eigenvalue of the form $\frac{2\var(X)}{\delta^2 + \tau}$ and combining Theorem \ref{thm:main} with \eqref{eq:debruijn} we deduce $O(1/n)$ convergence of relative entropy, using the fact that 
$$ \int_0^\infty \frac{1}{1+ \tau} \frac{1}{1 + (n-1)(C + t D)} d\tau = \frac{ \log \left( C/D + 1/(D(n-1)) \right)}{1+ (C-D) (n-1)} = O(1/n)$$
for $C > D$.

\section*{Acknowledgements}

The author would like to thank Professor Thomas Courtade of the Department of Electrical Engineering and Computer Sciences, University of California, Berkeley for extremely helpful discussions regarding this work, and for numerous pointers to relevant papers in the literature. I would also like to thank Professor Venkat Anantharam of the Department of Electrical Engineering and Computer Sciences, University of California, Berkeley for valuable suggestions concerning the maximal correlation. The idea to consider the eigenfunctions in the maximal correlation problem grew out of a Twitter conversation with Dr James V Stone, Honorary Reader in Vision and Computational Neuroscience at the University of Sheffield. The author would like to thank the Associate Editor, and three anonymous referees for their close reading of this paper and extremely helpful suggestions.

\end{document}